\documentclass[conference,letterpaper]{IEEEtran}

\usepackage[utf8]{inputenc} 
\usepackage[T1]{fontenc}
\usepackage{url}
\usepackage{ifthen}
\usepackage{cite}
\usepackage[cmex10]{amsmath} 

\usepackage{comment}
\RequirePackage{amsthm,amssymb,mathrsfs,amsfonts}
\usepackage{enumitem}  
\usepackage{euscript}       
\usepackage{color}
\usepackage{bm,bbm}
\usepackage{soul}
\usepackage{algorithm}
\usepackage{algpseudocode}
\usepackage{array}
\usepackage[caption=false,font=normalsize,labelfont=sf,textfont=sf]{subfig}
\usepackage{textcomp}
\usepackage{stfloats}
\usepackage{verbatim}
\usepackage{graphicx}
\usepackage{tikz}
\usetikzlibrary{positioning,arrows.meta}

\usepackage{balance}

\usepackage{amssymb}

\newtheorem{theorem}{Theorem}

\newtheorem{lemma}{Lemma}

\newtheorem{assumption}{Assumption}

\theoremstyle{remark}
\newtheorem{remark}{Remark}

\DeclareMathOperator*{\essup}{ess\,sup}

\newcommand{\Exp}{\mathsf{E}}
\newcommand{\Pro}{\mathsf{P}}
\newcommand{\WADD}{\mathrm{WADD}}

\newcommand{\cF}{\mathscr{F}}

\newcommand{\I}{\mathsf{
KL}_{M}}

\newcommand{\Tc}{\mathsf{T}}
\newcommand{\e}{\mathrm{e}}
\newcommand{\hth}{\hat{\theta}}

\renewcommand{\L}{\EuScript{L}}

\newcommand{\algname}{AQuTe-CUSUM}

\usepackage[
  letterpaper,
  left=0.64in,
  right=0.64in,
  top=0.76in,
  bottom=1.1in
]{geometry}

\begin{document}
\title{Quickest Change Detection in Parametric Models With 1-Bit Measurements} 


\author{%
  \IEEEauthorblockN{Liyan Xie}
  \IEEEauthorblockA{Department of Industrial and Systems Engineering \\
                    University of Minnesota\\
                   Minneapolis, Minnesota, USA\\
                    Email: liyanxie@umn.edu}
  \and
  \IEEEauthorblockN{Martina Cardone}
  \IEEEauthorblockA{Department of Electrical and Computer Engineering\\ 
  University of Minnesota\\
                   Minneapolis, Minnesota, USA\\
                    Email: mcardone@umn.edu}
}

\maketitle

\begin{abstract}
We consider the quickest change detection from 1-bit quantized observations, where the post-change distribution is a parametric model with unknown parameters and the quantization thresholds are jointly chosen with the detection statistic. We propose an {Adaptive-Quantization-Threshold CUSUM} (\algname{}) algorithm, which estimates the post-change parameter online and adaptively selects the quantization threshold to maximize the induced Kullback-Leibler divergence. Under suitable regularity conditions, we characterize the average run length and worst-case average detection delay of the \algname{} procedure, and show that it is asymptotically optimal in first-order as the average run length goes to infinity.
Finally, we assess the performance of \algname{} for two distributions, namely the Gaussian and Poisson. 
\end{abstract}

\section{Introduction}\label{sec:intro}
Online change-point detection is a fundamental problem in statistics and signal processing~\cite{poor-hadj-QCD-book-2008,Siegmund1985,tartakovsky2014sequential,tutorial_jsait}. The most common version of the problem consists of a sequence of observations sampled independently. There is a change-point such that the underlying distribution changes from one distribution to an alternative. This problem is of major importance in many applications, such as seismic detection~\cite{xie2019asynchronous}, industrial quality control~\cite{shi2009quality}, dynamical systems~\cite{lai1995sequential}, structural health control~\cite{balageas2010structural}, event detection~\cite{li2017detecting}, anomaly detection \cite {chandola2009anomaly}, and detection of attacks~\cite{tartakovsky2014rapid}. The goal of online change-point detection is to detect the occurrence of the change in statistical behavior with a minimal delay while controlling the false alarm rate. 


In systems constrained by power, bandwidth, or hardware, such as large-scale sensor networks or embedded monitoring devices, the observed data is typically a quantized version of the raw samples~\cite{shah2025generalized}. Quantization, in these cases, reduces communication overhead and memory usage, making it well-suited for resource-limited environments. Sequential analysis under quantized observations has been studied in some classical work~\cite{stein2018asymptotic,willett1995performance,blum2002quantization,lee1981sequential,tantaratana1977quantization,nguyen2006optimal}, with most of them focusing on given or fixed quantizers for every sample. In a recent work~\cite{fauss2022optimal}, the authors study the optimal data-adaptive quantizers for the online hypothesis testing problem with the goal of minimizing a weighted summation of expected stopping time and error probabilities. Another related line of work studies online change detection under controlled sensing constraints, where sensing actions are adapted using recent parameter estimates~\cite{veeravalli2024quickest}. While quantization can be viewed as a possible type of sensing action, it also brings features not directly considered in~\cite{veeravalli2024quickest}, such as continuously selected 1-bit thresholds rather than a finite action set, as well as estimation properties specific to 1-bit measurements.

In this work, we study online change-point detection from quantized data, where the quantization is in the form of 1-bit. We assume streaming data from a parametric distribution arrives sequentially at a sensing node, which applies a 1-bit quantizer to each incoming raw observation.
We assume a known pre-change raw data distribution and an unknown post-change data distribution. To improve detection performance, the sensing node is allowed to adapt its quantization threshold over time. Our goal is to minimize the detection delay (i.e., the time needed to detect the change) under a fixed false alarm constraint, by designing a rule that updates both the quantization threshold and the detection statistic over time.

We propose an {Adaptive-Quantization-Threshold CUSUM} (\algname{}) algorithm. The algorithm estimates the post-change distribution parameter in an online manner, using the most recent quantized observations, and then optimizes the next quantization threshold based on this estimate. Given a new quantized observation, we adopt the cumulative sum (CUSUM) algorithm to update the detection statistic efficiently by accumulating the log-likelihood ratios \cite{page-biometrica-1954,lai-ieeetit-1998}. We provide theoretical guarantees in terms of the average run length and worst-case average detection delay, and show that the proposed method is first-order asymptotically optimal.


The paper contributions and outline are as follows:
\begin{enumerate}[leftmargin=*, itemsep=0pt, topsep=0pt, parsep=0pt, partopsep=0pt]
\item Section~\ref{sec:setup} presents the problem setup and preliminaries. In particular, Theorem~\ref{thm:lower} provides a lower bound on the minimal detection delay incurred by using 1-bit measurements.
\item Section~\ref{sec:MainResult} presents the proposed algorithm and the main theoretical results of our work, that is, the average run length in Lemma~\ref{lem:arl} and detection delay in Theorem~\ref{thm:wadd}.
\item Section~\ref{sec:Examples} shows the applicability of our results through two parametric examples on Gaussian and Poisson models. 
\end{enumerate}




\section{Problem Setup and Preliminaries}\label{sec:setup}

Suppose that we have access to a univariate i.i.d. data sequence $\{X_t, t\in \mathbb{N}\}$ with $X_t \in \mathcal{X} \subseteq \mathbb{R}$. We assume that each $X_t$ follows a parametric family distribution with density $p_X(x;\theta)$ parameterized by $\theta\in\Theta \subset \mathbb{R}$, where $\Theta$ denotes the parameter space. 
There are two parameters $\theta_0$ and $\theta_1$, and a deterministic time $\kappa \in \{0,1,2,\ldots\}$ such that
\begin{equation}
X_t  \stackrel{\rm i.i.d.}{\sim}\left\{\begin{array}{ll}
p_X(x;\theta_0) ,&t = 1,2,\ldots,\kappa,\\[2pt]
p_X(x;\theta_1) ,&t =\kappa+1,\kappa+2,\ldots
\end{array}\right.
\label{eq:data_model}
\end{equation}
That is, $\kappa$ is a change-point, where the observations are i.i.d.~before and including $\kappa$ following the density $p_X(x;\theta_0)$, while for times after $\kappa$ they are i.i.d. following $p_X(x;\theta_1)$. We assume that the pre-change distribution parameter $\theta_0$ is known (as it can typically be estimated from abundant nominal data), but $\theta_1$ is {\em unknown}. Moreover, we assume that the post-change parameter $\theta_1$ belongs to a compact set $\Theta_1\subset \Theta$ and $\theta_0 \notin \Theta_1$, so that the detection problem is well-defined. 

Each data point $X_t$ is not observed directly, and instead the measurements are given by a 1-bit  mechanism, namely
\begin{equation}\label{eq:1bit}
B_{t} =
\begin{cases}
1 & \text{if } X_{t} \le \tau_t, \\
-1 & \text{if } X_{t} > \tau_t,
\end{cases}
\end{equation}
where $\tau_t$ is a threshold that may be chosen adaptively (e.g., based on $\{B_k\}_{k=1}^{t-1}$), subject to the constraint that $\tau_t$ lies in the interior of $\mathcal{X}$. Here, we impose the assumption that $\tau_t$ is chosen from a fixed bounded interval $\tau_t\in[\tau_{\min},\tau_{\max}]$ with $\tau_{\min} \in \mathcal{X}$ and $\tau_{\max} \in \mathcal{X}$.  
Thus, for $b\in\{-1,1\}$ and for $X_t\sim p_X(\cdot;\theta)$, we have the distribution of $B_t$ resulting from~\eqref{eq:1bit} as
\begin{equation}
\label{eq:pB}
p_{B} \left (b;\theta,\tau_t \right ) = \Pr \left (B_t=b \right ) = \int_{\mathcal{X}_t(b)} p_X(x;\theta)\,{\mathrm{d}} x,
\end{equation}
where $\mathcal{X}_t(b)$ is a subset of $\mathcal{X}$ determined by $b$; specifically,
$\mathcal{X}_t(1)= \left \{x\in\mathcal{X}:x\le\tau_t \right \}$ and  
$\mathcal{X}_t(-1)= \left \{x\in\mathcal{X}:x>\tau_t \right \}$.

From these observed data samples $\{B_t,t\in\mathbb{N}\}$, we seek to detect the unknown change-point $\kappa$ as quickly as possible. 
We note that, when different (but fixed, non-adaptive) thresholds are used, the $B_t$'s are still independent, but no longer identically distributed. 
%
We consider joint quantization-and-detection procedures. A joint procedure is a pair $(\pi,\Tc)$, where (i) $\pi=\{\tau_t\}_{t\ge1}$ is a quantization policy such that each $\tau_t\in[\tau_{\min},\tau_{\max}]$ is $\cF_{t-1}$-measurable, and (ii) $\Tc$ is a stopping time with respect to the filtration ${\cF_t}$, at which we stop and declare that a change has occurred before $\Tc$. Here, $\cF_t=\sigma\{B_1,\ldots,B_t\}$, with $\cF_0$ denoting the trivial sigma-algebra. In this work, we aim to design a joint strategy that sequentially selects each threshold based on past information and, after observing $B_t$, updates the detection statistic and determines whether to stop.

We denote with $\Pro_\infty^\pi,\Exp_\infty^\pi$ the probability measure and the corresponding expectation when all samples follow the pre-change distribution (i.e.,~$\kappa=\infty$), $\Pro_0^{\theta,\pi},\Exp_0^{\theta,\pi}$ when all data are under the post-change density with parameter $\theta$ (the change happens at 0), and with $\Pro_\kappa^{\theta,\pi},\Exp_\kappa^{\theta,\pi}$ the probability measure and expectation induced when the change happens at time $\kappa$.

We adopt two commonly used performance metrics: (i) the Average Run Length (ARL) {$\Exp^\pi_\infty[{\Tc}]$,} and (ii) Lorden's Worst-case Average Detection Delay (WADD) defined as $\WADD(\Tc)=\sup_{\kappa\geq 0}\,\essup {\Exp_{\kappa}^{\theta,\pi}}[(\Tc-\kappa)^{+}\mid \cF_\kappa]$, where $\essup$ denotes the essential supremum. Ideally, we would like to find  
the optimal joint procedure $(\pi,\Tc)$ that solves the following constrained optimization problem:
\begin{equation*}
\inf_{\Tc,\pi} \, \WADD(\Tc),\quad  \text{subject~to:}~{\Exp^\pi_\infty}[\Tc]\geq\gamma>1,
\label{eq:optCUSUM}
\end{equation*}
i.e., the optimal $\Tc$ has the smallest WADD among all stopping times {and quantization policies with} ARL no smaller than~$\gamma$.

We now provide a lower bound on the WADD for any detection procedure under 1-bit measurements satisfying the ARL constraint. For each threshold value $\tau$ and $\theta\in\Theta_1$, we define the 1-bit Kullback-Leibler (KL) divergence as
\begin{equation}
\label{eq:KL}
\mathsf{KL}(\theta,\tau)
:=\sum_{b\in\{-1,1\}} p_B(b;\theta,\tau)
\log\frac{p_B(b;\theta,\tau)}{p_B(b;\theta_0,\tau)}, 
\end{equation}
which is the KL divergence between the distributions of the quantized 1-bit measurement induced by $p_X(\cdot;\theta)$ and by the pre-change distribution $p_X(\cdot;\theta_0)$. 
\begin{theorem}[Lower Bound~\cite{lai-ieeetit-1998}]\label{thm:lower}
    As $\gamma \rightarrow \infty$, we have 
\begin{equation}
\label{eq:LBT}
\inf_{(\pi,\Tc): \Exp_\infty^\pi[\Tc] \geq \gamma} \WADD(\Tc)  \geq (\I^{-1} + o(1))\log\gamma,
\end{equation}
where {$\I = \max_{\tau\in[\tau_{\min},\tau_{\max}]}\mathsf{KL}(\theta_1,\tau)$}
is the maximum KL divergence between the pre- and post-change distribution of 1-bit measurements $B_t$ over all possible thresholds. 
\end{theorem}
\begin{proof}
The proof follows directly from~\cite{lai-ieeetit-1998}, {by noting that the key condition in~\cite{lai-ieeetit-1998} is that, in our notation, for every $\epsilon>0$, $\sup_{\pi}\sup_{\kappa\ge0} \essup \Pro_{\kappa}^{\theta_1,\pi}\big(\max_{1\le m\le n}\sum_{t=\kappa+1}^{\kappa+m}\Lambda_t^\pi>(1\!+\!\epsilon) \, n \, \I \! \mid \!\cF_\kappa\big)\!\to\!0$, where $\Lambda_t^\pi=\log{p_B(B_t;\theta_1,\tau_t)/p_B(B_t;\theta_0,\tau_t)}$. Since $\tau_t$ is $\cF_{t-1}$-measurable, $\Exp_{\kappa}^{\theta_1,\pi}[\Lambda_t^\pi\mid\cF_{t-1}]=\mathsf{KL}(\theta_1,\tau_t)\le\I$ for any $\tau_t$. Thus, the key condition in~\cite{lai-ieeetit-1998} continues to hold and the same lower bound holds over joint procedures $(\pi,\Tc)$.}
\end{proof}
Theorem~\ref{thm:lower} means that for any detection procedure, together with a suitable 
choice of the threshold, with ARL no smaller than $\gamma$, the minimal detection delay is of the order of $\log\gamma/\I$. Therefore, a detection procedure is called {\it first-order asymptotically optimal} if its WADD equals $\frac{\log\gamma}{\I}(1+o(1))$ asymptotically as $\gamma\rightarrow\infty$. In the next section, we propose a detection rule together with an adaptive criterion for selecting the thresholds $\tau_t$'s and prove that it achieves the lower bound in Theorem~\ref{thm:lower} asymptotically as $\gamma \to \infty$.

\section{Proposed Method}\label{sec:MainResult}

\begin{algorithm}[t]
\caption{\algname{}}
\label{alg:ab_cusum}
\begin{algorithmic}[1]
\Require Window size $w$, detection threshold $\nu$.
\State Initialize ${\mathsf{S}}_1=\cdots={\mathsf{S}}_w = 0$, and $\tau_1,\ldots,\tau_w$ randomly.
\State Observe $B_1,\ldots,B_w$ according to $\tau_1,\ldots,\tau_w$. 
\For{$t = w+1, w+2, \ldots$}
    \State Estimate $\hat{\theta}_{t-1}$ via~\eqref{eq:mle_window}
    \State Select quantization threshold $\tau_t$ via~\eqref{eq:threshold_selection}
    \State Observe $B_t$ via the 1-bit mechanism in~\eqref{eq:1bit}
    \State Update statistic ${\mathsf{S}}_t$ via~\eqref{eq:statAWLCUSUM}
    \If{${\mathsf{S}}_t \ge \nu$}
        \State \textbf{Stop} and declare a change at time $t$
    \EndIf
\EndFor
\end{algorithmic}
\end{algorithm}

We {assume that} the post-change parameter $\theta_1$ is unknown. We propose to adopt a sliding window to estimate $\theta_1$ {sequentially}. Fix a window size $w\ge1$, and for $t>w$, let 
\begin{equation}\label{eq:mle_window}
\hth_{t-1}
=\arg\max_{\theta\in {\Theta_1}}\sum_{s=t-w}^{t-1}\log p_{B}(B_s;\theta,\tau_s),
\end{equation}
that is, the maximum likelihood estimate based on the most recent $w$ observations $\{B_s\}_{s=t-w}^{t-1}$, where each $B_s$ is the 1-bit observation from~\eqref{eq:1bit} with threshold $\tau_s$.  Since $\hth_{t-1}$ depends only on past and current observations, it is $\cF_{t-1}$-measurable.

At time $t>w$, we choose the threshold $\tau_t$ as the one that maximizes the 1-bit KL divergence in~\eqref{eq:KL} under the current parameter estimate $\hth_{t-1}$:
\begin{equation}\label{eq:threshold_selection}
\tau_t = \arg\max_{\tau \in [\tau_{\min},\tau_{\max}]} \mathsf{KL} \left (\hth_{t-1},\tau \right ).    
\end{equation}
Thus, $\tau_t$ adapts to the recent data through $\hth_{t-1}$.
Then, we adopt the window-limited CUSUM procedure and update our detection statistic via the following recursion~\cite{xie2023window}:
\begin{equation}
{\mathsf{S}}_t
={\mathsf{S}}_{t-1}^+
+\log\frac{p_B(B_t;\hth_{t-1},\tau_t)}{p_B(B_t;\theta_0,\tau_t)}, t>w;
\,
{\mathsf{S}}_1=\cdots={\mathsf{S}}_w=0.
\label{eq:statAWLCUSUM}
\end{equation}
The corresponding stopping time is
\begin{equation}
{\Tc}
=\inf\{t>w:{\mathsf{S}}_t\ge \nu\},
\label{eq:stAWLCUSUM}
\end{equation}
where $\nu$ is a pre-specified detection threshold. We call the procedure \algname{} and summarize it in Algorithm~\ref{alg:ab_cusum}. 
We now analyze the ARL and the WADD of \algname{}. {To simplify the notation, we suppress the dependence on the quantization policy $\pi$ in~\eqref{eq:threshold_selection} used by \algname{}.}
\begin{lemma}[ARL of \algname{}]\label{lem:arl}
 For threshold $\nu=\log\gamma$, the stopping time ${\Tc}$ satisfies {$\Exp_\infty[{\Tc}] > \gamma.$}
\end{lemma}
\begin{proof}
The proof follows directly from~\cite{xie2023window}. Define the one-step log-likelihood ratio
\begin{equation}\label{eq:ell}
\ell_t = \log  \frac{p_B(B_t;\hth_{t-1},\tau_t)}{p_B(B_t;\theta_0,\tau_t)} ,    
\end{equation}
and the Shiryaev-Roberts-type statistic
\begin{equation}
\label{eq:SRStati}
\L_t=(\L_{t-1}+1)\exp(\ell_t), t>w;\, {\L_1 = \ldots =\L_w=0.}
\end{equation}
Since $\hth_{t-1}$ and $\tau_t$ are $\cF_{t-1}$-measurable, under $\Pro_\infty$ we have
\begin{equation*}
\Exp_\infty[\exp(\ell_t)\mid \cF_{t-1}]
=\!\!\!\!\!\sum_{b\in \{-1,1\}} \!\!\!\!\!p_B(b;\theta_0,\tau_t)
\frac{p_B(b ;\hth_{t-1},\tau_t)}{p_B(b;\theta_0,\tau_t)}=1.
\end{equation*}
Therefore, from~\eqref{eq:SRStati}, we arrive at
\begin{equation*}
\Exp_\infty[\L_t-t \mid \cF_{t-1}]
=\L_{t-1}-(t-1),
\end{equation*}
and so $\{\L_t-t\}$ is a martingale under $\Pro_\infty$.
Also, by exponentiating~\eqref{eq:statAWLCUSUM}, we arrive at
\begin{equation*}
\e^{{\mathsf{S}}_t}
=\max \left \{\e^{{\mathsf{S}}_{t-1}},1 \right \}\exp(\ell_t), t>w;\,
{ \e^{{\mathsf{S}}_1} = \ldots =\e^{{\mathsf{S}}_w}=1,}
\end{equation*}
which, by induction, implies that $\L_t\ge \e^{{\mathsf{S}}_t}$ for all $t>w$.
Note that $\Exp_\infty[\L_w-w]=-w<0$ and hence, by the optional sampling theorem, we have $\Exp_\infty[\L_{{\Tc}}-{\Tc}]<0$, which implies
\begin{equation*}
\Exp_\infty[{\Tc}]
> \Exp_\infty[\L_{{\Tc}}]
\ge \Exp_\infty[\e^{{\mathsf{S}}_{{\Tc}}}]
\ge \e^\nu
=\gamma.
\end{equation*}
This concludes the proof of Lemma~\ref{lem:arl}.
\end{proof}


\begin{lemma}[Worst-Case Delay for \algname{}]\label{lem:worst-case}
For any change-point $\kappa\geq0$ we have that
\begin{align}
\begin{split}
&\mathrm{ess}\sup\Exp_\kappa^\theta[(\Tc-\kappa)^{+}\mid \cF_\kappa] \\
&\leq \sup_{c_1,\ldots,c_w} \Exp_0^\theta[\Tc \mid  \tau_{1}=c_1,\ldots,\tau_{w}=c_w].  
\end{split}
\end{align}
\end{lemma}
\begin{proof}
Consider first a change at $\kappa=0$. Since we start testing at $w+1$, {we have that $\Tc > w$ and hence,} we can write
\begin{equation*}
\Exp_0^\theta[\Tc]=w+\Exp_0^\theta\Big[\Exp_0^\theta[\Tc-w\mid \cF_w]\Big],
\end{equation*}
{where the expectation $\Exp_0^\theta$ is taken with respect to $\{B_i\}_{i \geq 1}$ and $\{\tau_i\}_{i\geq 1}$.}
Suppose now that the change happens at $\kappa$. Then, for $t>\kappa+w$ it is clear that the test statistic $\mathsf{S}_t$ is larger than the test statistic $\mathsf{S}_{t,\kappa}$ generated by \textit{starting} the detection procedure at time $\kappa+w+1$.  
This suggests that, if we use $\mathsf{S}_{t,\kappa}$ instead of $\mathsf{S}_t$, then we will stop at a time $\Tc_\kappa$ that satisfies $\Tc_\kappa\geq\Tc$. With these observations in mind, we can write
\begin{align*}
&\Exp_\kappa^\theta[(\Tc-\kappa)^{+}\mid \cF_\kappa] 
\leq w+\Exp_\kappa^\theta[(\Tc-\kappa-w)^{+}\mid \cF_\kappa]\\
& \leq w+\Exp_\kappa^\theta[(\Tc_\kappa-\kappa-w)\mid \cF_\kappa]\\
& =\Exp_\kappa^\theta\Big[ \Exp_\kappa^\theta[(\Tc_\kappa-\kappa)\mid \tau_{\kappa+1},\ldots,\tau_{\kappa+w},\cF_\kappa]\mid \cF_\kappa\Big]\\
& \overset{{\rm{(a)}}}{=}\Exp_\kappa^\theta\Big[ \Exp_\kappa^\theta[(\Tc_\kappa-\kappa)\mid \tau_{\kappa+1},\ldots,\tau_{\kappa+w}]\mid \cF_\kappa\Big]\\
&\leq \sup_{c_1,\ldots,c_w} \Exp_\kappa^\theta[{\Tc_\kappa}-\kappa\mid \tau_{\kappa+1}=c_1,\ldots,\tau_{\kappa+w}=c_w]\\
& \overset{{\rm{(b)}}}{=} \sup_{c_1,\ldots,c_w} \Exp_0^\theta[\Tc \mid  \tau_{1}=c_1,\ldots,\tau_{w}=c_w],
\end{align*}
where $\Exp_\kappa^\theta$ is taken with respect to $\{B_i\}_{i > \kappa}$ and $\{\tau_i\}_{i\geq 1}$. 
The equality ${\rm{(a)}}$ is due to the fact that $\Tc_\kappa$ is independent from $\cF_\kappa$ given the thresholds $\{\tau_{\kappa+i}\}_{i=1}^w$ because in $\mathsf{S}_{t,\kappa}$ the data $\cF_\kappa$ are not being used, {and ${\rm{(b)}}$} is due to the definition of $\Tc_\kappa$ and $\Tc=\Tc_0$ when $\kappa=0$ and to the fact that, once the change happens at time $\kappa$, the data from that point onward {evolves as a fresh} process started at time $0$ under the post-change law.
\end{proof}
Lemma~\ref{lem:worst-case} implies that to upper bound the WADD of \algname{}, we only need to consider the special case when $\kappa=0$. 
We first define some functions that we will use frequently in the analysis. Define $q(\theta,\tau):=p_B(1;\theta,\tau)$, and, for $t>w$, $l_s(\theta):=\log p_B(B_s;\theta,\tau_s)$ and $L_{t-1,w}(\theta):=\frac1w\sum_{s=t-w}^{t-1} l_s(\theta)$.
Then,~\eqref{eq:mle_window} can be written as $\hat\theta_{t-1}=\arg\max_{\theta\in\Theta_1} L_{t-1,w}(\theta)$.


\begin{assumption}
{We assume that:}
\begin{enumerate}[label=(A\arabic*),leftmargin=*, labelsep=0.5em] 
\item $\Theta_1\subset\mathbb R$ is compact and $\theta_1$ belongs to the interior of $\Theta_1$.


\item The map $q(\theta,\tau)$ is jointly three-times continuously differentiable in
    $(\theta,\tau)$ on $\Theta_1 \times [\tau_{\min},\tau_{\max}]$, and all first-, second-, and third-order partial derivatives
are uniformly bounded on $\Theta_1\times[\tau_{\min},\tau_{\max}]$.

\item There exists $\delta\in(0,1/2)$ such that $\delta \le q(\theta,\tau)\le 1-\delta$, for any $(\theta,\tau)\in \Theta_1\times[\tau_{\min},\tau_{\max}]$.

\item For any $\tau\in[\tau_{\min},\tau_{\max}]$ and every $\varepsilon>0$, there exists $c_\varepsilon>0$ such that $\inf_{{|u-\theta|}\ge \varepsilon}
\mathsf D(\theta,u;\tau)\ge c_\varepsilon$, where 
\begin{equation}\label{eq:D}
 \mathsf D(\theta,u;\tau)
:=
\sum_{b\in\{-1,1\}}
p_B(b;\theta,\tau)
\log\frac{p_B(b;\theta,\tau)}{p_B(b;u,\tau)}\ge0.   
\end{equation}

\item 
There exists a constant $I_0>0$ such that $\forall \tau\in[\tau_{\min},\tau_{\max}]$
\begin{equation*}
\mathcal I_B(\theta_1,\tau):=
\frac{(\partial_\theta q(\theta_1,\tau))^2}{q(\theta_1,\tau)\{1-q(\theta_1,\tau)\}} \ge I_0.
\end{equation*}
\item For each $\theta$ in a neighborhood of $\theta_1$, the maximizer
    \begin{equation*}
    \tau^\star(\theta)\in \arg\max_{\tau\in[\tau_{\min},\tau_{\max}]}
    \mathsf{KL}(\theta,\tau)
    \end{equation*}
    is unique and interior, and $\partial_{\tau\tau}^2 \mathsf{KL}(\theta_1,\tau^\star(\theta_1)) < 0$.
\end{enumerate}    
\end{assumption}

{(A1)-(A4) are standard 
conditions ensuring consistency of the 1-bit MLE~\cite[Thm 5.7]{vandervaart1998}. (A5) ensures an $O(w^{-1})$ mean-square estimation rate, while (A6)} ensures stability of the KL-optimal quantization threshold.
These conditions hold, for example, for regular one-parameter location families with a positive smooth density when the parameter and threshold sets are compact; the Gaussian location model in Section~\ref{sec:Gaussian} is one such example. For the Poisson model in Section~\ref{sec:Poisson}, when parameter estimation is restricted to a compact interval, (A3)-(A5) are satisfied. Although (A2) and (A6) do not apply directly because the threshold set is discrete, the KL-optimal cutoff is unique at $\theta_1$ and remains optimal when the parameter estimate is sufficiently close to $\theta_1$, so the same first-order asymptotic conclusion holds for the finite threshold set used here. Developing a unified analysis that accommodates both 
continuous and discrete threshold sets is left for future work.

\begin{theorem}[WADD of \algname{}]\label{thm:wadd}
Suppose Assumptions (A1)--{(A6)} hold. Assume threshold $\nu=\log\gamma$, window size $w=o(\log\gamma)$ but $w\to\infty$ as $\gamma\to\infty$. We have the following upper bound for the WADD of \algname{}:
\begin{equation}\label{eq:th1}
\begin{aligned}
\WADD(\Tc) \leq\frac{\log\gamma}{\I}(1+o(1)), 
\end{aligned}
\end{equation}
where ${\I} = \max_{\tau \in [\tau_{\min},\tau_{\max}]} \mathsf{KL}(\theta_1,\tau)$. Thus, the \algname{} procedure is {first-order} asymptotically optimal. 
\end{theorem}

\begin{proof} 
From Lemma~\ref{lem:worst-case}, we have that the worst-case delay is attained at $\kappa=0$, i.e., we only need to upper bound $\Exp_0^\theta[\Tc \mid \tau_{1}=c_1,\ldots,\tau_{w}=c_w]$ for any $c_1,\ldots,c_w$. Therefore, we fix $c_1,\ldots,c_w$ and work under the probability measure 
$\Pro_0^{\theta_1}(\,\cdot\mid \tau_1=c_1,\ldots,\tau_w=c_w)$, and for simplicity we omit the conditioning on $c_1,\ldots,c_w$ in the remaining of this proof.

\vspace{0.05in}
\noindent {\it Step 1. Consistency of $\hth_{t-1}$.}
We first show that, as $w\to\infty$, $\hth_{t-1}$ is a consistent estimate of $\theta_1$ and $\Exp_0^{\theta_1}[(\hat\theta_{t-1}-\theta_1)^2]
\le C/w$ for some positive constant $C$.
Conditioning on $\cF_{s-1}$, the threshold $\tau_s$ is deterministic, hence for any fixed $u\in\Theta_1$, 
\begin{equation*}
\Exp_0^{\theta_1}\big[l_s(u)-l_s(\theta_1) \mid \cF_{s-1}\big]
=
-\mathsf D(\theta_1,u;\tau_s) \le 0,
\end{equation*}
where $\mathsf D$ is defined in~\eqref{eq:D}. 
From Assumption~(A3), $|l_s(u)-l_s(\theta_1)|$ is bounded, thus by the weak law of large numbers for bounded martingale differences (see, e.g.,~\cite[Ch.~2]{hallheyde1980}),
\begin{equation*}
L_{t-1,w}(u)-L_{t-1,w}(\theta_1)
\xrightarrow{\Pro_0^{\theta_1}}  
-\frac1w\sum_{s=t-w}^{t-1}\mathsf D(\theta_1,u;\tau_s).
\end{equation*}
Now, by Assumption~(A4), for any $\varepsilon>0$, there exists $c_\varepsilon>0$ such that for any $u$ satisfying
$|u-\theta_1|\ge \varepsilon$, $L_{t-1,w}(u)-L_{t-1,w}(\theta_1)$ is negative with probability tending to one.
Since $\hat\theta_{t-1}$ maximizes $L_{t-1,w}(u)$ over
$\Theta_1$, it follows that $\hat\theta_{t-1}$ is a consistent estimate of $\theta_1$.
Moreover, by a Taylor series expansion on the function $L_{t-1,w}$, we have 
\begin{equation*}
0=\sum_{s=t-w}^{t-1}\partial_\theta l_s(\theta_1) + \left (\sum_{s=t-w}^{t-1}\partial_{\theta\theta}^2 l_s(\tilde\theta_{t-1} )\right )  (\hat\theta_{t-1}-\theta_1 ) ,
\end{equation*}
for some $\tilde\theta_{t-1}$ between $\hat\theta_{t-1}$ and $\theta_1$. {By (A5), we have that $-\frac{1}{w}\sum_{s=t-w}^{t-1}\partial_{\theta\theta}^2 l_s(\tilde\theta_{t-1})$ is lower bounded by a positive constant with probability tending to one and, for all sufficiently large~$w$,}
\begin{equation}
\label{eq:BoundExpDiffTheSquare}
\Exp_0^{\theta_1}\big[(\hat\theta_{t-1}-\theta_1)^2\big]\le \frac{C}{w}.
\end{equation}

\vspace{0.05in}
\noindent {\it Step 2. Delay Upper Bound.}
For $\ell_t$ defined in~\eqref{eq:ell}, we let $U_t:=\sum_{s=w+1}^t \ell_s$ for $t>w$, {$U_t=0$ for $t \leq w$,} and define
\begin{equation*}
\Tc^\circ:=\inf\{t>w:U_t\ge \nu\}.
\end{equation*}
Since $\mathsf S_t\ge U_t$, we have $\Tc\le \Tc^\circ$. Therefore, it suffices to prove that the upper bound in~\eqref{eq:th1} holds for $\Exp_0^{\theta_1}[\Tc^\circ]$. 

Since $\{\ell_t,t\ge1\}$ is a sequence of random variables adapted to $\{\mathcal{F}_{t},t\ge 1\}$, by Proposition~1 in~\cite{veeravalli2024quickest}, to prove that~\eqref{eq:th1} is an upper bound on $\Exp_0^{\theta_1}[\Tc^\circ]$, it suffices to prove~that:

{
\noindent $\bullet$ The lower bound of $\Exp_0^{\theta_1}\left[\ell_t|\mathcal F_{t-w-1}\right]$ converges to $\I$ almost surely. We start by noting that}
\begin{equation}\label{eq:expansion}
\begin{aligned}
\Exp_0^{\theta_1}\left[\ell_t|\mathcal F_{t-1}\right] &=
\!\!\sum_{b \in \{-1,1\}}p_B(b;\theta_1,\tau_t)
\log\frac{p_B(b;\hat\theta_{t-1},\tau_t)}{p_B(b;\theta_0,\tau_t)}\\
& = \mathsf{KL}(\theta_1,\tau_t)-\mathsf{D}(\theta_1,\hat\theta_{t-1};\tau_t).
\end{aligned}
\end{equation}
By (A2)-(A3), the KL function is  uniformly bounded; hence, by a second-order Taylor expansion in $\mathsf{D}(\theta_1,\hat\theta_{t-1};\tau_t)$, together with the uniform boundedness of the derivatives and the fact that the first-order term vanishes at \(\theta_1\), we obtain
\begin{equation}\label{eq:part1}
\left|
\mathsf{D}(\theta_1,\hat\theta_{t-1};\tau_t)
\right|
\le
C(\hat\theta_{t-1}-\theta_1)^2,    
\end{equation}
for some positive constant $C$.
Also, by (A6), $\tau_t=\tau^\star(\hat\theta_{t-1})$, and the argmax map
$\theta\mapsto \tau^\star(\theta)$ is locally Lipschitz in a compact neighborhood of $\theta_1$.
Hence,
\begin{equation*}
|\tau_t-\tau^\star|\le C_2|\hat\theta_{t-1}-\theta_1|,
\end{equation*}
for some positive constant $C_2$. Since $\tau^\star$ is an interior maximizer of $\tau\mapsto \mathsf{KL}(\theta_1,\tau)$, we have
$\partial_\tau \mathsf{KL}(\theta_1,\tau^\star)=0$, and a second-order Taylor expansion yields
\begin{equation}\label{eq:part2}
\mathsf{KL}(\theta_1,\tau^\star)-\mathsf{KL}(\theta_1,\tau_t)
\le C_3|\tau_t-\tau^\star|^2
\le C_4(\hat\theta_{t-1}-\theta_1)^2, 
\end{equation}
for some positive constants $C_3$ and $C_4$.
Then, we obtain
\begin{equation}\label{eq:expansion2}
\begin{aligned}
&\Exp_0^{\theta_1}\left[\ell_t|\mathcal F_{t-w-1}\right] = \Exp_0^{\theta_1}\left[ \Exp_0^{\theta_1}[\ell_t|\mathcal F_{t-1}] \mid \mathcal F_{t-w-1}\right] \\
& {\overset{{\rm{(a)}}}{=}}  \Exp_0^{\theta_1} \left[\mathsf{KL}(\theta_1,\tau_t)-\mathsf{D}(\theta_1,\hat\theta_{t-1};\tau_t) \mid \mathcal F_{t-w-1}\right] \\
& \overset{{\rm{(b)}}}{\ge }\mathsf{KL}(\theta_1,\tau^\star) - C \Exp_0^{\theta_1} \left[ (\hat\theta_{t-1}-\theta_1)^2 |\mathcal F_{t-w-1}\right]\\
& \overset{{\rm{(c)}}}{\ge } \I - O\left(\frac{1}{w}\right),
\end{aligned}
\end{equation} 
where the labeled (in)equalities follow from: $\rm{(a)}$ using~\eqref{eq:expansion}; ${\rm{(b)}}$ using~\eqref{eq:part1} and~\eqref{eq:part2}; and ${\rm{(c)}}$ using~\eqref{eq:BoundExpDiffTheSquare}. Therefore, the lower bound of $\Exp_0^{\theta_1}\left[\ell_t|\mathcal F_{t-w-1}\right]$ converges to $\I$ as $w\to\infty$. 

\noindent $\bullet$ The upper bound of $\Exp_0^{\theta_1}\left[\ell_t|\mathcal F_{t-w-1}\right]$ is $\I$. This follows directly from~\eqref{eq:expansion} and~\eqref{eq:expansion2}.

\noindent $\bullet$ The upper bound of $\Exp_0^{\theta_1}\left[\ell^2_t|\mathcal F_{t-w-1}\right]$ is finite almost surely. {This follows directly from (A3).}
\end{proof}

\begin{remark}
{Consistency of the estimator $\hth_{t-1}$ is critical to the proof of 
Theorem~\ref{thm:wadd}. This requirement may not hold for some cases; e.g.,} if $X_t\sim\mathcal N(0,\sigma^2)$ and {$\tau_t = 0$,} then {$\Pr(B_t=1)=1/2$} for every $\sigma$ {((A4) and (A5) are violated),} so no estimator based only on the resulting bits can consistently estimate $\sigma$. In practice, $\sigma$ can still be consistently estimated if informative nonzero thresholds are used sufficiently often.
\end{remark}

\begin{remark}
{When $\theta_1$ is a vector
we may} need to use multiple $\tau$-values to guarantee that the Fisher information matrix $\mathcal I_B(\theta_1,\tau)$ is positive definite. For instance, consider $X\sim \mathcal{N}(\mu,\sigma^2)$. If we only use a fixed constant $\tau$, then observing $\{B_t,t\in\mathbb{N}\}$ is not sufficient to estimate both $\mu$ and $\sigma$. Thus, for vector values of $\theta_1$, we 
need to have multiple distinct $\tau$ values 
to ensure that the true 
parameter can be identified. 
\end{remark}


\section{Examples} \label{sec:Examples}

{In this section, we assess the numerical performance of \algname{} for two common distributions.}

\subsection{Gaussian Distribution {$X_t \sim \mathcal N(\theta,1)$}}\label{exp:normal}
\label{sec:Gaussian}
Suppose that before {the change $X_t\sim\mathcal N(\theta_0,1)$,} while after the change $X_t\sim \mathcal N(\theta_1,1)$. 
{From~\eqref{eq:pB}, for any 
$\tau$,} we have $p_B(1;\theta,\tau)=\Pr(X_t\le \tau)=\Phi(\tau-\theta)$, where $\Phi(\cdot)$ is the standard normal cumulative distribution function. Hence, from~\eqref{eq:KL}, the 1-bit KL divergence~is
\begin{align}
\begin{split}
\mathsf{KL}(\theta_1,\tau)
&= \Phi(\tau-\theta_1)
\log\frac{\Phi(\tau-\theta_1)}{\Phi(\tau-\theta_0)}
\\
&\quad +
\bigl(1-\Phi(\tau-\theta_1)\bigr)
\log\frac{1-\Phi(\tau-\theta_1)}{1-\Phi(\tau-\theta_0)}.
\end{split}
\label{eq:KLGaussian}
\end{align}

\begin{figure}[t]
    \centering
    \begin{tabular}{cc}
    \includegraphics[width=0.46\linewidth]{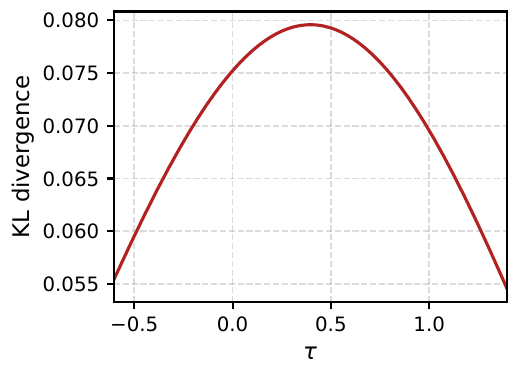}  & \includegraphics[width=0.46\linewidth]{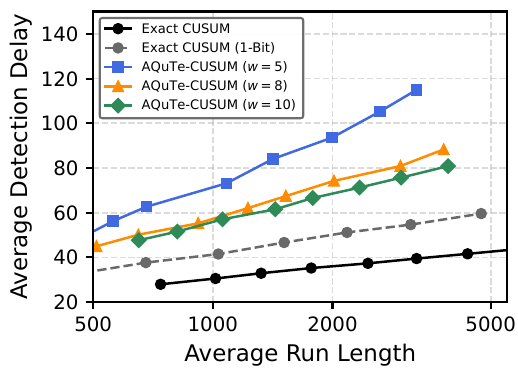}    
       \end{tabular}
       \vspace{-0.2in}
    \caption{$\mathcal{N}(\theta,1)$, with $\theta_0=0$, $\theta_1=0.5$, and $\Theta_1 = [0.25, \infty)$. 
    }
    \label{fig:guassian}
    \vspace{-0.10in}
\end{figure}

{The optimal $\tau^\star$ is the value that maximizes~\eqref{eq:KLGaussian}.} In general, $\tau^\star$ does not admit a simple closed-form expression. In practice, however,
$\mathsf{KL}(\theta_1,\tau)$ is a smooth one-dimensional function of $\tau$, so $\tau^\star$ can be computed efficiently numerically. 
{Fig.~\ref{fig:guassian} (left) illustrates $\mathsf{KL}(\theta_1,\tau)$ versus $\tau$;}  Fig.~\ref{fig:guassian} (right) shows that the delay of our method {has similar slope as} the information-theoretical lower bound obtained by the 1-bit exact CUSUM test, {especially for a moderate window size such as $w=10$.}


\subsection{Poisson Distribution {$X_t\sim \mathsf{Poi}(\theta)$}}\label{exp:poisson}
\label{sec:Poisson}
Suppose that before the change $X_t\sim\mathsf{Poi}(\theta_0)$ and after the change $X_t\sim\mathsf{Poi}(\theta_1)$ with $\theta_1\neq \theta_0$. 
Since the support is discrete, the threshold $\tau$  induces the event $\{X_t\le \tau\}$ to be equivalent to $\{X_t\le \lfloor\tau\rfloor\}$. Defining $k:=\lfloor\tau\rfloor$, we have that
{$p_B(1;\theta,\tau)
=
\sum_{x=0}^{k} {\rm{e}}^{-\theta}\frac{\theta^x}{x!}.$}
\begin{figure}[t]
    \centering
    \begin{tabular}{cc}
    \includegraphics[width=0.46\linewidth]{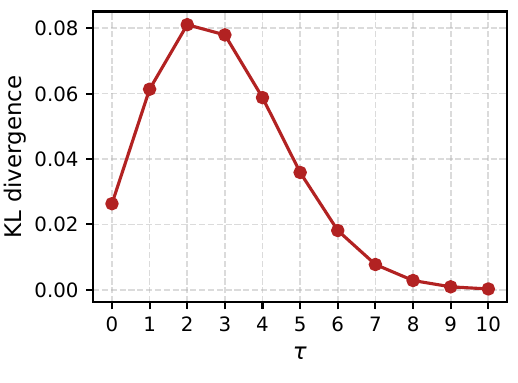}  & \includegraphics[width=0.46\linewidth]{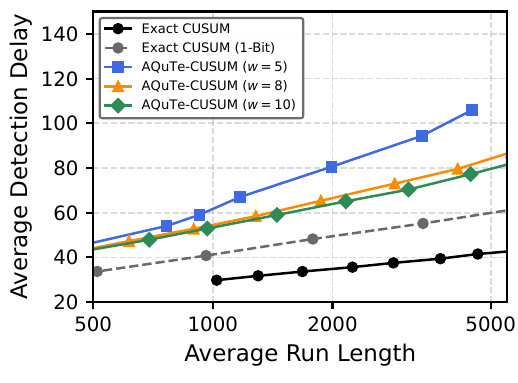}    
       \end{tabular}
       \vspace{-0.2in}
    \caption{$\mathsf{Poi}(\theta)$, with $\theta_0=2$, $\theta_1=2.75$, and $\Theta_1 = [2.275, \infty)$. 
    }
    \label{fig:poisson}
    \vspace{-0.12in}
\end{figure}
{Unlike the Gaussian case, here the optimization is 
over the integers in $[\tau_{\min},\tau_{\max}]$. 
Thus, in practice, 
$\tau^\star$} can be obtained by direct enumeration over all admissible integer cutoff values $\tau\in \{\lceil \tau_{\min}\rceil,\ldots,\lfloor \tau_{\max}\rfloor\}$.
{Fig.~\ref{fig:poisson} (left) illustrates this observation and Fig.~\ref{fig:poisson} (right) shows the average detection delay versus ARL, from which we can draw similar conclusions as for the Gaussian case.}


\section{Conclusion and Discussion}
{We studied 
quickest} change detection based on 1-bit quantized observations, where the post-change distribution follows a parametric model with unknown parameters, and the quantization thresholds are designed jointly with the detection statistic. We proposed an algorithm, \algname{}, and showed that it achieves the smallest worst-case average detection delay in first-order as the average run length goes to infinity.
As part of future work, we plan to consider: (i) the practical implementation of threshold selection. 
{The exact maximizer $\tau_t$ may not be available in closed-form  and hence, it would be interesting} 
to understand how to develop an efficient approximation algorithm and analyze the corresponding impact on detection delay. (ii) In practice, there might be {some
constraints,} such as $|\tau_{t} - \tau_{t-1}|\le \delta$ {for 
$\delta>0$} or $\tau_t$ can only be updated once for every $m$ time steps. 
Thus, it is left as future work to develop a modified detection algorithm and analyze its theoretical performance under 
these constraints.

\newpage
\balance
\bibliographystyle{IEEEtran}
\bibliography{references.bib}

\end{document}